%% file: priority_queue.tex
\documentclass[11pt, letterpaper]{article}

\usepackage{fullpage,amsmath,amsthm,amssymb}
\usepackage[margin=1in]{geometry}

\newcommand{\D}{\mathcal{Q}}

\newcommand{\Dist}{D}
\newcommand{\UINTa}[1]{\mathcal{U}^{\mathrm{SI}}_{{#1}}}
\newcommand{\SIa}[1]{\mathsf{SetInt}_{{#1}}}

\renewcommand{\log}{\lg}

\newcommand{\Probes}{\mathcal{P}}
\newcommand{\Source}{\mathcal{L}}
\newcommand{\Dest}{\mathcal{R}}

\newcommand{\Tree}{\mathcal{T}}

\DeclareMathOperator{\E}{\mathbb{E}}
\DeclareMathOperator{\Var}{\mathrm{Var}}
\newcommand{\eps}{\varepsilon}

\newtheorem{innercustomlem}{Restatement of Lemma}
\newenvironment{customlem}[1]
  {\renewcommand\theinnercustomlem{#1}\innercustomlem}
  {\endinnercustomlem}

\newtheorem{observation}{Observation}
\newtheorem{property}{Property}
\newtheorem{lemma}{Lemma}
\newtheorem{theorem}{Theorem}
\newtheorem{definition}{Definition}

\title{DecreaseKeys are Expensive for External Memory Priority Queues}

\author{Kasper Eenberg\thanks{MADALGO,
	Center for Massive Data Algorithmics, 
	a Center of the Danish National Research Foundation, grant DNRF84.
	Department of Computer Science, Aarhus University.
	$\{$\texttt{kse,larsen}$\}$\texttt{@cs.au.dk}. Larsen is also
        supported by a Villum Young Investigator Grant and an AUFF Starting Grant.}
 \and 
 	Kasper Green Larsen\footnotemark[1] \and Huacheng Yu\thanks{Department of Computer Science, Stanford University. \texttt{yuhch123@gmail.com}. Supported by NSF CCF-1212372.}}
\date{}

\begin{document}
\maketitle
\begin{abstract} 
One of the biggest open problems in external memory data structures is
the priority queue problem with DecreaseKey operations. If only Insert
and ExtractMin operations need to be supported, one can design a
comparison-based priority queue performing $O((N/B)\lg_{M/B} N)$ I/Os
over a sequence of $N$ operations, where $B$ is the disk block size in
number of words and $M$ is the main memory size in number of
words. This matches the lower bound for comparison-based sorting and
is hence optimal for comparison-based priority queues. However, if we
also need to support DecreaseKeys, the performance of the best known
priority queue is only $O((N/B) \lg_2 N)$ I/Os. The big open question
is whether a degradation in performance really is necessary. We answer
this question affirmatively by proving a lower bound of $\Omega((N/B)
\lg_{\lg N} B)$ I/Os for processing a sequence of $N$ intermixed
Insert, ExtraxtMin and DecreaseKey operations. Our lower bound is
proved in the cell probe model and thus holds also for
non-comparison-based priority queues.
%We complement our lower bound by a new non-comparison-based priority queue achieving $O((N/B) \lg_{\lg N} N)$ I/Os. 
\end{abstract}
\thispagestyle{empty}
\newpage 
\setcounter{page}{1}
\input{intro}
\input{commModel}
\input{pqLowerBounds}

\input{com_lower}
\bibliographystyle{abbrv}
\bibliography{bibliography}
\appendix
\input{appendix}

\end{document}

%% file: intro.tex
\section{Introduction}
The priority queue is one of the most fundamental data structures. Priority queues support maintaining a set $S$ of keys with associated priorities. The three most commonly supported operations on a priority queue are the following:
\begin{itemize}
\item \textbf{Insert}($k,p$): Add the key $k$ to $S$ with associated priority $p$.
\item \textbf{ExtractMin}: Remove and return the (key-priority)-pair $(k,p)$ from $S$, such that $p$ is minimal amongst all priorities associated to keys in $S$. If multiple elements have the same smallest priority, an arbitrary of these pairs may be returned.
\item \textbf{DecreaseKey}($k,p$): Given a (key,priority)-pair $(k,p)$ such that there is already a pair $(k,p')$ present in $S$, change the priority associated with $k$ to the minimum of $p$ and $p'$.
\end{itemize}
The classic solution to the problem is the binary heap~\cite{cormen:introduction} data structure, which supports all three operation in $O(\lg N)$ time, where $N$ is the size of $S$. This data structure assumes only that priorities can be compared and thus is optimal by reduction from comparison-based sorting. Another classic solution is the Fibonacci heap~\cite{Fredman:fibo} which supports the DecreaseKey in amortized $O(1)$ time per operation, Insert in worst case $O(1)$ time and ExtractMin in $O(\lg N)$ amortized time. A paper by Brodal~\cite{Brodal:pq} shows how to make all these operations run in worst case time as well. If we assume priorities and keys are integers bounded by a polynomial in $N$, we can get below the comparison-based sorting lower bound. In fact, Thorup~\cite{Thorup:2007:equiv} showed that priority queues and integer sorting are equivalent in this setting, and thus by using the currently fastest integer sorting algorithms~\cite{HanTho,Han:sort}, one obtains a priority queue with either expected amortized $O(\sqrt{\lg \lg N})$ time for each operation, or deterministic amortized $O(\lg \lg N)$ time per operation. One important thing to remark for integer keys (or even just keys that can be stored in hash tables with $O(1)$ expected amortized time bounds), is that we can always support the DecreaseKey operation in the same time bound as Insert and ExtractMin if we allow Las Vegas randomization and amortization in all three bounds. For the interested reader, see the reduction in Section~\ref{sec:internaldecrease}. In addition to the equivalence to sorting, Thorup~\cite{thorup:integerdecrease} has also presented a deterministic priority queue for integer priorities which supports constant time Insert and DecreaseKey, while supporting ExtractMin in $O(\lg \lg N)$ time.

\paragraph{External Memory Data Structures.}
In many applications, in particular in database systems, the maintained data set $S$ is too large to fit in the main memory of the machine. In such cases, most of the data structure resides in slow secondary memory (disk). In this setting, classic models of computation such as the word-RAM are very poor predictors of performance. This is because random access to disk is several orders of magnitude slower than random access in RAM, and thus memory accesses are far more expensive than CPU instructions. The external memory model of Aggarwal and Vitter~\cite{av-iocsr-88} was defined to make better predictions of actual performance of algorithms when data resides on disk. In the external memory model, a machine is equipped with an infinitely sized disk and a main memory that can hold only $M$ elements. The disk is partitioned into blocks that can store $B$ elements each. A data structure in this model can only perform computation on data in main memory. A data structure transfers data between external memory and main memory in blocks of $B$ elements. The transfer of a block is called an I/O. The query time and update time of a data structure is measured only in the number of I/Os spend and thus computation on data in main memory is free of charge. The motivation for moving blocks of $B$ elements in one I/O, is that reading or writing $B$ consecutive locations on disk is essentially just as fast as a single random access, even when $B$ is in the order of a million elements. Thus the goal in external memory data structures is to make sure that memory accesses are very local.

Note that in external memory, the cost of scanning $N$ elements on disk is only $O(N/B)$ I/Os (and not $O(N)$). Another important bound is the sorting bound, which says that $N$ comparable keys can be sorted in $O(\textrm{Sort}(N))=O((N/B)\lg_{M/B}(N/M))$ I/Os~\cite{av-iocsr-88} and this is optimal. A natural assumption is that $N \geq M^{1+\eps}$ for some constant $\eps>0$, in which case the bound simplifies to $O((N/B)\lg_{M/B} N)$. For readability of the bounds, we make this assumption in the remainder of the paper. The relationship between $M$ and $B$ is such that we always require $M \geq 2B$. This ensures that the main memory can hold at least two blocks of data. In the literature, it is often assumed that $M \geq B^2$, which is known as the \emph{tall-cache} assumption.

\paragraph{Priority Queues in External Memory.}
The priority queue is also one of the most basic data structures in the external memory model of computation. By taking the classic binary heap and modifying it to an $M/B$-ary tree with buffers of $M$ elements stored in each node~\cite{Fadel:heaps}, one can support the Insert and ExtractMin operation in amortized $O((1/B)\lg_{M/B} N)$ I/Os per operation. For comparison based priority queues, this bound is optimal by reduction from sorting. Now the really interesting thing is that if we are also to support the DecreaseKey operation, the best known priority queue, due to Kumar and Schwabe~\cite{Kumar:tournament} has all three operations running in amortized $O((1/B)\lg_2 N)$ I/Os per operation. This is in sharp contrast to classic models of computation where one can always support DecreaseKey at no extra cost if allowed Las Vegas randomization and amortization. Moreover, the classic Fibonacci heap even has the DecreaseKey running in $O(1)$ time. 

In addition to priority queues being interesting in their own right, they also have immediate implications for Single Source Shortest Paths (SSSP) in external memory. The original paper of Kumar and Schwabe introduced their priority queue in order to obtain more efficient algorithms for external memory SSSP. They show that solving SSSP on a graph $G=(V,E)$ can be done in $O(|V|)$ I/Os plus $O(|E|)$ Inserts, $O(|E|)$ DecreaseKeys and $O(|E|)$ ExtractMin operations on an external memory priority queue. In their solution, it is crucial that the DecreaseKey operation has precisely the interface we defined earlier (i.e. taking a key and a new priority as argument without knowing the previous priority).
Combined with their priority queue, their reduction gives the current fastest external memory SSSP algorithm for dense graphs, running in $O(|V| + (|E|/B)\lg_2 N)$ I/Os. Thus faster priority queues with DecreaseKeys immediately yield faster external memory SSSP algorithms for dense graphs (see reference below for the sparse graph case). The big open question, which has remained open for at least the 20 years that have past since~\cite{Kumar:tournament}, is whether or not a degradation in performance really is necessary if we want to support DecreaseKeys.

 Our main result is an affirmative answer to this question:
\begin{theorem}
\label{thm:main}
Any, possibly Las Vegas randomized, external memory priority queue supporting Insert in expected amortized $t_I$ I/Os, ExtractMin in expected amortized $t_E$ I/Os and DecreaseKey in expected amortized $t_D$ I/Os, must have $\max\{t_I,t_E,t_D\} = \Omega((1/B) \lg_{\lg N} B)$ when $N \geq M^{4+\eps}$. The lower bound holds also for non-comparison-based priority queues.
\end{theorem}
When $B \geq N^{\eps}$, this lower bound is $\Omega((1/B)\lg_{\lg N} N)$, whereas under the tall-cache assumption ($M \geq B^2$), one can design priority queues with only Insert and ExtractMin with an amortized cost of $O(1/B)$ per operation when $B \geq N^{\eps}$ (since $\lg_{M/B}N$ becomes $O(1)$). This is a $\lg N/\lg \lg N$ factor difference in performance!  Note that when we are proving lower bounds for non-comparison-based data structures, we assume keys and priorities are $O(\lg N)$-bit integers and that a disk block holds $Bw$ bits, where $w = \Theta(\lg n)$ is the word size. Similarly, we assume the main memory can hold $Mw$ bits. Our lower bound thus rules out an $O(|V|+\textrm{Sort}(|E|))$ I/O SSSP algorithm directly from a faster priority queue. Our lower bound also shows an interesting separation between internal memory and external memory models by showing that Las Vegas randomization and amortization is not enough to support DecreaseKeys at no extra cost in the external memory model.

%We complement our lower bound with a new non-comparison-based priority queue:
%\begin{theorem}
%\label{thm:ds}
%There is a Las Vegas randomized external memory priority queue that supports Insert, ExtractMin and DecreaseKey in expected amortized $O((1/B) \lg_{\lg N} N)$ I/Os per operation when keys and priorities are integers bounded by a polynomial in $N$.
%\end{theorem}
%This leaves only a gap of $\lg N/\lg B = \lg_B N$ between the upper and lower bound. Which is closest to the truth remains an open problem. At least, a lower bound of $\Omega((1/B)\lg_{\lg N} N)$ is not possible without assuming $B \geq \lg^c N$ for some sufficiently large constant $c>0$, since running Thorup's (internal memory) priority queue with $O(\sqrt{\lg \lg N})$ time operations is faster than $O((1/B)\lg_{\lg N} N)$ when $B = o(\lg N/(\lg \lg N)^{3/2})$.
%
%The proof of Theorem~\ref{thm:ds} can be found in Section~\ref{sec:upperbound}.

\paragraph{Related Work.}
Recent work of Wei and Yi~\cite{wei:equiv} proves an external memory analog of Thorup's result on equivalence between integer sorting and priority queues. More specifically, they show that if sorting $N$ integers can be done in $O(S(N))$ I/Os, then there is an external memory priority queue supporting Insert and ExtractMin in amortized $O(S(N)/N)$ I/Os per operation, provided that $S(N) = \Omega(2^{\lg^* N})$ where $\lg^*N$ is the iterated logarithm. If $S(N) = o(2^{\lg^* N})$, their reduction gives a priority queue with Insert and ExtractMin in amortized $O((S/N)\lg^* N)$ I/Os. Note that we are dealing with integers and not just comparable keys and priorities, thus we have no sorting lower bound and it might be the case that $S(N)=O(N/B)$. In fact, Wei and Yi mention that their reduction might be useful for proving a lower bound for external memory sorting of integers via a lower bound for priority queues. Unfortunately, their reduction does not give efficient DecreaseKeys, and thus our result does not imply an integer sorting lower bound.

For the related dictionary problem, i.e. maintaining a set $S$ of integers under membership queries, a sequence of papers by Yi and Zhang~\cite{yi:member}, Verbin and Zhang~\cite{verbin:member} and finally Iacono and P{\v a}tra{\c s}cu~\cite{patrascu12buffer}, established tight lower bounds in external memory. Note that all these bounds are for maintaining integers and thus comparison based lower bounds do not apply. Moreover, the upper bounds for maintaining integers are far better than the comparison-based lower bounds say is possible for comparable keys and priorities, see~\cite{Brodal:member}.

We also mention that in some applications where one needs to perform DecreaseKey operations, one might know the previous priority of the (key,priority)-pair being updated. In that case, one can support all three operations in $O((1/B)\lg_{M/B}N)$ I/Os per operations using e.g. the reduction to sorting due to Wei and Yi~\cite{wei:equiv}.

Finally, for the application to SSSP mentioned early, there are more efficient solutions for sparse graphs. The best known algorithm for sparse graphs is due to Meyer and Zeh~\cite{Meyer2006} and solves SSSP in roughly $O(\sqrt{|E||V|/B} \log_2 |V|)$ I/Os (plus some I/Os for computing an MST). For $|E| > |V|B$, this is slower than the solution of Kumar and Schwabe.

%% file: commModel.tex
\section{Two-Phase Communication Complexity}
\label{sec:introtwophase}
As mentioned in the introduction, our lower bound for priority queues hold also for non-comparison-based priority queues. More specifically, we think of the main memory as being able to store $M$ words of $w = \Theta(\lg N)$ bits each. Similarly, a block has a total of $Bw = \Theta(B \lg N)$ bits. We assume disk blocks have integer addresses in $[2^w]$. Priority queues can now perform arbitrary computations on the main memory free of charge and we only count the number of disk blocks read and written to. This corresponds exactly to the cell probe model~\cite{Yao1981:tables} with cell size $Bw = \Theta(B \lg N)$ bits, augmented with a main memory of $Mw = \Theta(M \lg N)$ bits that is free to examine. To use a terminology consistent with the literature on cell probe lower bounds, we henceforth refer to reading or writing a disk block as a cell probe (a cell probe thus corresponds to an I/O).

Our lower bound for priority queues follows by a reduction from communication complexity. Such an approach is not at all new in data structure lower bounds, see e.g.~\cite{MNSW:1998}. However for the priority queue problem, the standard reductions yield communication games with solutions that are too efficient to yield lower bounds for priority queues. We thus exploit that the protocols we obtain have a very special structure. To explain this further, we first introduce the communication game we reduce from and then discuss the known upper bounds that forces us to look at the special structure of our protocols:

Consider a communication game in which two players Alice and Bob both
receive a subset of a universe $[U]$. Letting $x$ denote Alice's set
and $y$ denote Bob's, the goal for the two players is that \emph{both players} learn
the intersection of the two sets $x$ and $y$, while minimizing
communication with Alice. This is the \emph{set intersection} problem in communication complexity. The special case of set intersection in which the players only need to decide whether $x \cap y = \emptyset$ is known as the \emph{set disjointness} problem and is amongst the most fundamental and well-studied problems in communication complexity. We assume the reader is familiar with basic terminology in communication complexity, i.e. deterministic protocols, private and public coin randomized protocols etc.

The Las Vegas, or zero-error, randomized communication complexity of set intersection has been shown to be $O(N)$ if Alice and Bob both receive sets of size $N$ in a universe $[U]$, i.e. independent of $U$~\cite{brody:setintersect}. This is too weak to give any useful lower bound for priority queues - we really need bounds higher than $N$. One could ask if using deterministic communication complexity would help, as the deterministic case has higher communication lower bounds. Unfortunately, our reduction from priority queues inherently yields Las Vegas randomized protocols even when then priority queue is deterministic (we use Yao's min-max principle). In addition to this issue, our application to priority queues have Alice's set somewhat smaller than Bob's, i.e. we are considering an asymmetric communication game, see e.g.~\cite{MNSW:1998}. Sadly there are also too efficient upper bounds for the asymmetric case to yield any lower bound for priority queues.  What saves us is that our reduction from set intersection to priority queues results in protocols of a very special form: First Bob speaks a lot while Alice is almost silent. They then switch roles and Alice speaks a lot while Bob is almost silent, i.e. protocols have two phases. They key property of external memory priority queues, which makes one of the players almost silent, is that if the data structure probes $k$ cells, then those cells have $kBw$ bits. Specifying the address of a cell is thus much cheaper than specifying its contents. This shows up in the communication protocol where one player essentially sends a factor $B$ less bits than the other. It turns out that such restrictions on the communication protocol rules out efficient solutions and we can in fact prove lower bounds that are higher than the known upper bounds without this two phase restriction. We define Two-Phase protocols and their cost in the following:

We say that a Las Vegas randomized communication protocol $P$ is a Two-Phase protocol if for every root-to-leaf path in the protocol tree, there is exactly one node which is marked the \emph{phase-transition} node. We think of all communication before the phase-transition node as belonging to phase one and all communication afterwards as belonging to phase two. We say that a Las Vegas randomized Two-Phase protocol $P$ has expected cost $C(P,x,y):=(a_1,b_1,a_2,b_2)$ on input $(x,y)$ ($x$ is given to Alice and $y$ to Bob) if: When running $P$ on $(x,y)$, the expected number of bits sent by Alice in phase one is $a_1$, the expected number of bits sent by Bob in phase one is $b_1$, the expected number of bits sent by Alice in phase two is $a_2$ and the expected number of bits sent by Bob in phase two is $b_2$. We use $C_{a_1}(P,x,y)$ as short for the expected communication of Alice in phase one. We define $C_{b_1}(P,x,y), C_{a_2}(P,x,y)$ and $C_{b_2}(P,x,y)$ similarly. If the inputs to Alice and Bob are drawn from a distribution $\mu$, we say that the cost of $P$ under distribution $\mu$ is $C(P,\mu):=(a_1,b_1,a_2,b_2)$ where $a_1 = \E_{(X,Y) \sim \mu}[C_{a_1}(P,X,Y)], b_1 = \E_{(X,Y) \sim \mu}[C_{b_1}(P,X,Y)], a_2 = \E_{(X,Y) \sim \mu}[C_{a_2}(P,X,Y)]$ and $b_2 = \E_{(X,Y) \sim \mu}[C_{b_2}(P,X,Y)]$.

Having defined Two-Phase protocol and their costs, we face one more challenge: Not any hard distribution for Two-Phase set intersection can be used in our reduction to priority queues. We are thus forced to prove lower bounds for a product distribution over Alice's and Bob's sets. Let $\SIa{U,k,l}$ denote the set intersection communication problem with sets in a universe $[U]$, such that Alice's set has size $k$ and Bob's set has size $l$. We present the hard distribution in the following and our main theorem lower bounding the communication complexity of set intersection under this distribution:

\paragraph{Hard distribution $\UINTa{U,k,l}$ for $\SIa{U,k,l}$}
Alice's input $X$ is sampled from all ${U\choose k}$ subsets of $[U]$ of size $k$ uniformly at random. Independently, Bob's input $Y$ is sampled from all ${U\choose l}$ subsets of $[U]$ of size $l$ uniformly at random. They must both learn $X\cap Y$.

\begin{lemma}\label{lem_com_lower_si}
No Las Vegas Two-Phase protocol $P$ for $\SIa{U,k,l}$ can have 
\[
	C(P,\UINTa{U,k,l})=(o(k),o(l),o(k\log r),o(k/r))
\]
for any $r>1$, $l\gg k$ and $U\gg l$.
\end{lemma}

Observe that the last bound (Alice's communication in the second phase) is $k \lg r$, even if Bob sends up to $o(l)$ bits in the first phase. If we did not have two phases, a lower bound saying that it is not possible for Bob to send $o(l)$ bits and Alice to send $o(k \lg r)$ bits is simply not true since for any $\eps>0$, there is a protocol in which Alice sends $k \lg (1/\eps)$ bits and Bob sends $\eps l$ bits. Thus it is absolutely crucial that we have two phases, in which Bob's longest message is sent in the first phase and Alice's longest message is sent in the second. We defer the proof of Lemma~\ref{lem_com_lower_si} to Section~\ref{sec:comm}. In the next section, we give our reduction to priority queues from Two-Phase protocols and derive our priority queue lower bound from Lemma~\ref{lem_com_lower_si}. 

For the observant reader, we make an interesting technical remark: The lower bound holds even for $r$ as large as $k$, even when $U$ is just slightly larger than $k$ and $l$. One could perhaps think that in this case, Alice could just send her set to Bob in the second phase for $O(k \lg(U/k))$ bits (and this is the only communication), which would then contradict the lower bound of $\Omega(k \lg r)$ if $\lg r = \omega(\lg(U/k))$. The crucial point is that such a protocol will not have Alice learn the intersection, so Bob has to send the intersection back. The expected size of the intersection is $k \ell /U$, so Bob's message would have to be that long. But this is already $\Omega(k/r)$ unless $U \geq \ell r \geq k r$. This contradicts $\lg r$ being $\omega(\lg(U/k))$.

%% file: pqLowerBounds.tex
\section{Two-Phase Set Intersection to Priority Queues}
Our goal is to prove a lower bound for external memory priority queues support Insert, ExtractMin and DecreaseKey. We actually prove a slightly stronger lower bound. More specifically, we prove lower bounds for priority queues that support Insert, ExtractMin and Delete. The Delete operation takes as input a key, and must remove the corresponding (key,priority)-pair from the priority queue. A priority queue supporting DecreaseKey rather than Delete can support Delete as follows: First do a DecreaseKey, setting the priority of the deleted key to $-\infty$ (or some small enough number). Then do an ExtractMin. We thus focus on proving lower bounds for priority queues with Insert, ExtractMin and Delete.

We let $B$ denote the disk block size in number of words of $w = \Theta(\lg N)$ bits and let $M$ denote the main memory size in number of words. To keep the math simple, we prove the lower bound only in the regime $N \geq M^{4+\eps}$ for some constant $\eps>0$. Our goal is to prove

\begin{theorem}
\label{thm:mainDelete}
Any Las Vegas randomized external memory priority queue with expected amortized $t_I$ probes per Insert, expected amortized $t_D$ probes per Delete and expected amortized $t_E$ probes per ExtractMin must have $\max\{t_I,t_D,t_E\} = \Omega((1/B) \lg_{\lg N} B)$ when $N \geq M^{4+\eps}$.
\end{theorem}
By the arguments above, this theorem also implies Theorem~\ref{thm:main}.
To prove the theorem, we start by describing our hard input distribution. In our hard distribution, we can have Deletes being executed for keys that have not been inserted in the data structure. For the remainder of our proof, we simply assume that a Delete operation for a key that is not present, has no effect. We show in Section~\ref{sec:avoidspurios} how we can modify the proof to avoid Deletes of keys not present in the data structure. Note also that our proof in the following only needs $N \geq M^{1+\eps}$. Altering the hard distribution such that we have no Deletes of already inserted keys needs the extra requirement $N \geq M^{4+\eps}$ as stated in the theorem.

\paragraph{Hard Distribution.}
Let $\beta \geq 2$ be an integer parameter. For an integer height $h \geq 8$ and memory size $M$, such that $N \geq M^{1+\eps}$ for some constant $\eps>0$, consider a $(2+\beta)$-ary tree $\Tree$ having sequences of updates at the leaves. We assume for convenience that $h$ is some multiple of $4$. Such a tree $\Tree$ defines a sequence of updates by forming the concatenated sequence of the updates in the leaves of $\Tree$ as encountered during a pre-order traversal of $\Tree$. 

The topology of $\Tree$ is defined recursively in the following way: Given a height parameter $h_v$, with $h_r = h$ for the root node $r$, a node $v \in \Tree$ is defined as follows: If $h_v = 0$, $v$ is simply a leaf. We call such a leaf a \emph{delete-leaf}. Otherwise, $v$ has $2+\beta$ children $c_1(v),c_2(v),\dots,c_{2+\beta}(v)$, in that order (from left to right). Children $c_1(v)$ and $c_{2+\beta}(v)$ are leaves, whereas $c_2(v),\dots,c_{1+\beta}(v)$ are recursively constructed each with height parameter $h_v-1$. We call $c_1(v)$ an \emph{insert-leaf} and $c_{2+\beta}(v)$ an \emph{extract-min-leaf}.

Having defined the topology of $\Tree$, we add updates to the delete-leaves of $\Tree$ as follows: We place $Mh$ Delete operations in each of the $\beta^h$ delete-leaves. The keys assigned to these Deletes are chosen by picking a uniform random set of $Mh\beta^h$ keys from $[(M h \beta^h)^4]$ and assigning them in uniform random order to the Delete operations.

For each internal node $v \in \Tree$ of height $h_v$, we place $M\beta^{h_v}$ Inserts in the insert-leaf $c_1(v)$. These Inserts all have the priority $h_v$. Summed over all insert-leaves, we have a total of $Mh\beta^h$ Inserts. We now assign keys to the Inserts by drawing a uniform random set of $Mh\beta^h$ keys from $[(M h \beta^h)^4]$ and assigning these keys in uniform random order to the Inserts.

Finally, we need to populate the extract-min leaves. For an internal node $v \in \Tree$ of height $h_v$, we place $M\beta^{h_v}$ ExtractMins in the extract-min leaf $c_{2+\beta}(v)$. These ExtractMins returns a sequence of key-priority pairs $(k_1,p_1),\dots,(k_{M \beta^{h_v}}, p_{M \beta^{h_v}})$, where this sequence is completely determined from previous operations in $\Tree$. We now re-insert all these pairs $(k_i,p_i)$ for which $p_i \neq h_v$.

This completes the description of the hard distribution. We refer to it as distribution $\Dist(\beta, h, M)$. We first state a simple property of the distribution:

\begin{property}
\label{prop:numups}
Distribution $\Dist(\beta ,h, M)$ is over sequences of $M h \beta^h$ Deletes, at most $2Mh \beta^h$ Inserts and $Mh\beta^h$ ExtractMins. 
\end{property}

The much more interesting property of $\Dist(\beta, h, M)$ is the following: Consider a node $v$ in $\Tree$ and define the set $Y_v$ consisting of all keys inserted in $c_1(v)$. The keys in $Y_v$ are all inserted with a priority of $h_v$ in $c_1(v)$. Also let $X_v$ denote the set of keys deleted in some delete-leaf in the subtrees rooted at $c_2(v),\dots,c_{1+\beta}(v)$. Precisely those keys in $Y_v$ that are not deleted by the Deletes corresponding to $X_v$, will still have priority $h_v$ just before processing $c_{2+\beta}(v)$. Moreover, all Inserts performed in the insert-leaves of the subtrees rooted at $v$'s children are with lower priorities. But these are all extracted in the corresponding extract-min leaves. Thus when reaching $c_{2+\beta}(v)$, the keys in $Y_v$ that have not been deleted are the keys in the priority queue with lowest priority. Therefore, the key-priority pairs being extracted by the ExtractMins in $c_{2+\beta}(v)$ uniquely determines $Y_v \setminus X_v$. Combined with $Y_v$,  we therefore also learn $Y_v \cap X_v$.

What this property suggests, is that the priority queue has to solve a set intersection problem between $Y_v$ and $X_v$ in order to correctly answer the ExtractMins in $c_{2+\beta}(v)$. But the keys in the delete-leaves are involved in $h$ such set intersection problems, and thus intuitively they will have to be read $h$ times. This suggests roughly the lower bound we aim for, namely that each delete operation causes $h/B$ I/Os. Of course this whole argument is complicated by the fact that if the priority queues figures out what the priorities are for the keys deleted by $X_v$, then the priority queue can ignore all elements in $X_v$ that do not correspond to keys of priority $h_v$ when solving the set intersection problem on $Y_v$ and $X_v$. We thus have to carefully argue that if the priority queue is very fast, then intuitively it cannot learn these priorities. Another technical issue we need to tackle, is that we have to consider set intersection between $Y_v$ and the set of Deletes in only one child's subtree. The formal details are given later.

We are now ready to set the stage for our reduction from the Two-Phase communication complexity of set intersection. The first step in the reduction is to find a suitable node $v$ in $\Tree$ into which to ``embed'' the set intersection communication problem. We show how we do this in the following subsection.

\subsection{Finding a Node for Embedding}
In the following, we show that for any efficient Las Vegas randomized external memory priority queue $\D^*$, if we run an update sequence drawn from $\Dist(\beta, h, M)$ on $\D^*$, then there is a node $v$ in $\Tree$, such that when $\D^*$ processes the updates in the subtree rooted $v$, there is only very few memory cells that are being updated during the operations corresponding to $c_1(v)$ and then read again during the operations in the subtree rooted at one of the other children $c_k(v)$. The precise details become apparent in the following.

\paragraph{Fixing Random Coins of Data Structures.}
Assume $\D^*$ is a Las Vegas randomized external memory priority queue with main memory size $M$, block size $B$, word size $w=\Theta(\lg N)$, having expected amortized $t_I$ probes per Insert, expected amortized $t_D$ probes per Delete and expected amortized $t_E$ probes per ExtractMin. By Property~\ref{prop:numups}, $\D^*$ processes an entire sequence of updates drawn from $\Dist(\beta,h,M)$ with an expected $t \leq 2Mh\beta^ht_I + Mh\beta^ht_D + Mh\beta^ht_E$ probes in total. Defining $N=4Mh\beta^h$ as the maximum total length of an update sequence drawn from $\Dist(\beta,h,M)$, we see that $t = O((t_I + t_D + t_E)N)$. By fixing the random coins of $\D^*$, we get a deterministic external memory priority $\D$ with main memory size $M$, block size $B$, word size $w$, using expected $t$ probes to process an update sequence drawn from $\Dist(\beta,h,M)$. We thus prove lower bounds on $t$ for such deterministic priority queues $\D$.

\paragraph{Assigning Probes to Nodes.}
Let $I \sim \Dist(\beta,h,M)$ be a random sequence of updates drawn from $\Dist(\beta,h,M)$. We use $\Tree$ to denote the $(2+\beta)$-ary tree corresponding to $I$. Given a deterministic external memory priority queue $\D$, let $t_\D(I)$ be the total number of probes performed by $\D$ when processing $I$. Our goal is to lower bound $\E[t_\D(I)]$. To do this, consider processing $I$ on $\D$. For each probe $\D$ performs, let $\ell$ be the leaf of $\Tree$ in which the currently processed update resides and let $a \in [2^w]$ be the address of the cell probed. If this is the first time the block of address $a$ is probed, we say the probe is associated to the leaf $\ell$. Otherwise, let $z$ be the leaf in which $a$ was last probed. We then associate the probe to the lowest common ancestor of $z$ and $\ell$. For a node $v \in \Tree$, define $\Probes(v)$ as the set of probes associated to $v$ after processing the entire update sequence $I$.

Since probes are associated to only one node each, it follows by linearity of expectation that there must exist a height $h^* \in \{h/2,\dots, h \}$ such that $$\E\left[\sum_{v \in \Tree : v \textrm{ is internal}, h_v=h^* } |\Probes(v)|\right] = O(t/h).$$
Consider the $\beta^{h-h^*}$ internal nodes at depth $h^*$ in $\Tree$. Define for each node $v$ of depth $h^*$ the quantity $\mathcal{C}(v)$ which gives the number of probes performed by all $O(\beta^{h^*}Mh)$ updates in the subtree rooted at $v$.

By Markov's inequality, linearity of expectation and a union bound, it follows that there is at least one of these nodes $v^* \in T$ with
$$
\E\left[|\Probes(v^*)|\right] = O\left(\frac{t}{h\beta^{h-h^*}}\right).
$$
and at the same time
$$
\E[\mathcal{C}(v^*)] = O(t/\beta^{h-h^*}).
$$
For a probe $p \in \Probes(v^*)$, let $i(p)$ and $j(p)$, with $i(p) < j(p)$, be the indices of the children $c_{i(p)}(v^*)$ and $c_{j(p)}(v^*)$ such that the probe $p$ took place in the subtree rooted at $c_{j(p)}(v^*)$ and was to a cell last probed in the subtree rooted at $c_{i(p)}(v^*)$. We define for each child $c_k(v^*)$ the sets of probes $\Source(v^*,k)$ and $\Dest(v^*,k)$, where $\Source(v^*,k)$ contains all probes $p \in \Probes(v^*)$ with $i(p) = k$ and $\Dest(v^*,k)$ contains all probes $p \in \Probes(v^*)$ with $j(p)=k$. Since each probe $p \in \Probes(v^*)$ is stored only in $\Source(v^*,i(p))$ and $\Dest(v^*,j(p))$, it follows that there must exist an index $k^* \in \{2,\dots,\beta+1\}$ such that
$$
\E\left[|\Source(v^*,k^*)| + |\Dest(v^*,k^*)|\right] = O\left( \frac{t}{h\beta^{h-h^*+1}}\right).
$$
Using that $t = O((t_I+t_D+t_E)N) = O(\max\{t_I,t_D,t_E\}M h \beta^h)$, the following lemma summarizes the findings above:
\begin{lemma}
\label{lem:embednode}
If $\D^*$ is a Las Vegas randomized external memory priority queue with main memory size $M$, block size $B$, word size $w$, expected amortized $t_I$ probes per Insert, expected amortized $t_D$ probes per Delete and expected amortized $t_E$ probes per ExtractMin, then there exists a deterministic external memory priority queue $\D$ with main memory size $M$, block size $B$, word size $w = \Theta(\lg N)$, such that for $I \sim \Dist(\beta, h, M)$, there is a node $v \in T$ with $h_v \in \{h/2,\dots,h\}$ and a child index $k \in \{2,\dots,\beta+1\}$ such that 
$$
\E\left[|\Source(v,k)| + |\Dest(v,k)|\right]  = O\left(\max\{t_I,t_D,t_E\}M \beta^{h_v-1}\right)
$$
and at the same time
$$
\E[\mathcal{C}(v)] = O(\max\{t_I,t_D,t_E\}Mh \beta^{h_v}).
$$
\end{lemma}
 The remaining part of the argument is to prove a lower bound on $\E\left[|\Source(v,h)| + |\Dest(v,h)|\right]$ and thereby a lower bound on $\max\{t_I,t_D,t_E\}$. We proceed in the following section to establish this lower bound by a reduction from Two-Phase communication complexity.

\subsection{Embedding the Communication Game}
With Lemma~\ref{lem:embednode} established, we are ready to show our reduction from Two-Phase communication complexity. 
Let $D^*$ be a Las Vegas randomized external memory priority queue as in Lemma~\ref{lem:embednode}. Then by Lemma~\ref{lem:embednode}, there is a deterministic priority queue $\D$ for which we can find a node $v \in T$ with $h_v \in \{h/2,\dots,h\}$ and a child index $k \in \{2,\dots,\beta+1\}$ such that
$$
\E\left[|\Source(v,k)| + |\Dest(v,k)|\right]  = O\left(\max\{t_I,t_D,t_E\}M \beta^{h_v-1}\right)
$$
and
$$
\E[\mathcal{C}(v)] = O(\max\{t_I,t_D,t_E\}Mh \beta^{h_v}).
$$
We will use this deterministic priority queue $\D$ and node $v$ to obtain a Two-Phase protocol for the set intersection problem in which Alice receives a set $X$ of size $M h \beta^{h_v-1}$ and Bob a set $Y$ of size $M\beta^{h_v}$. Assume these sets are drawn from the distribution $\UINTa{(M h \beta^h)^4,Mh\beta^{h_v-1},M\beta^{h_v}}$, i.e. the sets are uniform random and independent of each other from the universe $[(M h \beta^h)^4] = [N^4]$. Alice and Bob will now form an update sequence to the priority queue such that the distribution of that update sequence is exactly $\Dist(\beta, h, M)$. Their high-level approach is as follows: Alice and Bob will use randomness to populate all leaves of $\Tree$ up to (and including) reaching the node $c_{2+\beta}(v)$ in a pre-order traversal of $\Tree$. The updates in $c_1(v)$ will be known only to Bob and will represent his set $Y$. The updates in the delete-leaves of $c_k(v)$ will be known only to Alice and will represent her set $X$. The updates in all other populated nodes will be known to both of them. We give the details of how these nodes are populated in the following:

\paragraph{Populating $\Tree$.}
First we focus on populating all delete-leaves in $\Tree$ encountered up to reaching node $c_{2+\beta}(v)$ in a pre-order traversal of $\Tree$, except those in the subtree $c_k(v)$. This is done as follows: Let $L \leq M h \beta^h$ denote the number of Delete operations that we need to choose keys for. From public randomness, Alice and Bob sample a uniform random set of $L$ keys from $[(M h \beta^h)^4]$. Alice privately checks whether the sampled set intersects $X$. If not, she sends a 1-bit to Bob. If they do intersect, she first sends a 0-bit to Bob. She then privately samples a new set of $L$ keys that is uniform random amongst all sets disjoint from $X$ and sends that set to Bob. From public randomness, they assign the sampled keys to the Delete operations in a uniform random order. 

Next, they populate all insert-leaves encountered up to $c_{2+\beta}(v)$ in a pre-order traversal of $\Tree$, except $c_1(v)$. Let $K \leq Mh \beta^h$ denote the number of Insert operations they need to choose keys for. From public randomness, they sample a uniform random set of $K$ keys from $[(M h \beta^h)^4]$. Bob privately checks whether the sampled set intersects $Y$. If not, Bob sends a 1-bit to Alice. If they do intersect, Bob sends a 0-bit to Alice. He then privately samples a new set of $K$ keys that is uniform random amongst all sets disjoint from $Y$ and sends that set to Alice. From public randomness, they assign the sampled keys to the Insert operations in a uniform random order.

Finally, Bob privately populates $c_1(v)$ by assigning the keys in his set $Y$ in a uniform random order to the Insert operations in $c_1(v)$. Alice privately populates the delete-leaves in the subtree rooted at $c_k(v)$ by assigning the keys in her set $X$ in a uniform random order to the Delete operations.

Observe that by the rejection sampling, the distribution of the Insert and Delete operations in the insert-leaves and delete-leaves of $\Tree$ correspond exactly to distribution $\Dist(\beta,h,M)$ when $X$ and $Y$ are drawn from $\UINTa{(M h \beta^h)^4,Mh\beta^{h_v-1},M\beta^{h_v}}$. Also observe that after having chosen these, the operations to be performed in each extract-min leaf encountered before $v$ in a pre-order traversal of $\Tree$ are completely fixed under distribution $\Dist(\beta,h,M)$. Thus we assume Alice and Bob also both know these operations. Furthermore, Bob can privately determine all operations to be performed in the extract-min leaves in the subtrees rooted at $c_2(v),\dots,c_{k-1}(v)$ if following distribution $\Dist(\beta, h, M)$. We are ready to describe the communication protocol.

\paragraph{The Two-Phase Communication Protocol.}
In the first phase, Alice and Bob first send the bits needed by the rejection sampling procedure described above. Next, from the sampling above, Alice and Bob both know all updates of all leaves in the sampled $\Tree$ encountered before $v$ in a pre-order traversal of $\Tree$. They both privately run the deterministic priority queue $\D$ on this update sequence. Next, Bob knows all updates in $c_1(v),\dots,c_{k-1}(v)$. He privately feeds these updates to the priority queue $\D$ and remembers the set of addresses $A$ of the cells changed by these updates. He sends the set of such addresses to Alice for $O(|A|w)$ bits. He also sends the memory contents of $\D$ for another $O(Mw)$ bits. Alice now starts running the updates of the subtree $c_k(v)$ (she knows all these operations, except for the Inserts in the extract-min leaves). When processing these updates, whenever they want to read a cell for the first time during the updates of $c_k(v)$, Alice checks if the cell's address is in $A$. If not, she knows the correct contents of the cell from her own simulation of the data structure up to just before $v$. Otherwise, she asks Bob for the contents of the block (costing $O(w)$ bits to specify the address) as it is just before processing $c_k(v)$ and Bob replies with its $O(B w)$ bit contents. From that point on, Alice will remember all changes to that cell. Observe that by this procedure, when reaching an ExtractMin leaf, she can first simulate the query algorithm of $\D$ on all the ExtractMin operations. From the answers to these, she knows which Inserts to perform in the extract-min leaf in order to follow distribution $\Dist(\beta, h, M)$, so she runs these Inserts as well. Once Alice has finished running all updates in the subtree $c_k(v)$, phase one ends. In phase two, Alice starts by collecting the set of addresses $Z$ of all cells changed by the updates in $c_k(v)$'s subtree. She sends this set of addresses to Bob for $O(|Z|w)$ bits. She also sends the memory contents of $\D$ after processing all updates in $c_k(v)$'s subtree. This costs $O(Mw)$ vits. Bob now starts running all updates in $c_{k+1}(v),\dots,c_{2+\beta}(v)$ on $\D$. Each time one of these updates reads a cell with an address in $Z$ for the first time, Bob asks Alice for the contents of the block and Alice replies. This costs $O(w)$ bits from Bob and $O(B w)$ bits from Alice. If a block not with an address in $Z$ is read, Bob knows the contents himself. When Bob runs the ExtractMin operations in $c_{2+\beta}(v)$, he remembers the returned key-priority pairs $(k_1,p_1),\dots,(k_{M h \beta^{k_v}}, p_{M h \beta^{k_v}})$. It follows from our distribution $\Dist(\beta,h,M)$ that $Y \cap X$ is precisely the set of keys that were inserted during $c_1(v)$, not deleted in any of the delete-leaves in the subtrees rooted at $c_2(v),\dots,c_{k-1}(v),c_{k+1}(v),\dots,c_{1+\beta}(v)$ and yet was not returned with priority $h_v$ by an ExtractMin operation in $c_{2+\beta}(v)$. But Bob knows all these things, so he can compute $Y \cap X$. He then sends this set to Alice for $O(|X \cap Y|\lg N)$ bits, finishing the communication protocol.

\paragraph{Analysis.}
What remains is to analyse the efficiency of the Two-Phase protocol above. First observe that $L$ and $K$ are both at most $N=Mh\beta^h$ and the same for $|X|$ and $|Y|$. Since the universe has size $N^4$, the probability of having to reject the sampled sets are less than $1/N^2$. So the expected communication due to rejection sampling is $O(1)$ bits for both players. Now since the distribution of the update sequence fed to $\D$ is $\Dist(\beta, h, M)$, it follows from Lemma~\ref{lem:embednode} that $\E[|A|w] \leq \E[\mathcal{C}(v) w] = O(\max\{t_I,t_D,t_E\}Mhw \beta^{h_v}) = O(\max\{t_I,t_D,t_E\}hw|Y|)$. Since $h_v \geq h/2$, we see that $O(Mw) = o(|Y|)$. Next, observe that the number of cells for which Alice needs to ask Bob for the contents in phase one is $|\Dest(v,k)|$. Thus the expected number of bits sent by Alice to ask for cell contents is $O(\max\{t_I,t_D,t_E\}M \beta^{h_v-1} w) =O(\max\{t_I,t_D,t_E\}M \beta^{h_v-1} w) = O(\max\{t_I,t_D,t_E\}|X|w/h)$. Bob's replies cost an expected 
$$
O(\max\{t_I,t_D,t_E\}M \beta^{h_v-1} B w) = O(\max\{t_I,t_D,t_E\}|Y|Bw/\beta)
$$
bits.

For phase two, sending the addresses of cells in $Z$ cost an expected 
$$
\E[|Z|w]\leq \E[\mathcal{C}(v) w] = O(\max\{t_I,t_D,t_E\}Mhw \beta^{h_v}) = O(\max\{t_I,t_D,t_E\}\beta w |X|)$$
bits. Since $h_v \geq h/2$, we have $O(Mw)=o(|X|)$. Next, the number of cells for which Bob needs to ask Alice for the contents is $|\Source(v,k)|$. Thus the expected number of bits sent by Bob to ask for cell contents is $O(\max\{t_I,t_D,t_E\}M \beta^{h_v-1} w) = O(\max\{t_I,t_D,t_E\}|Y|w/\beta)$. The replies of Alice costs expected $O(\max\{t_I,t_D,t_E\}MB \beta^{h_v-1} w) = O(\max\{t_I,t_D,t_E\}Bw|X|/h)$ bits. Finally, the expected size of $|X \cap Y|$ is $O(1/N^2)$ bits, thus the last message from Bob to Alice (sending $X \cap Y$) costs and expected $O(1)$ bits.

\paragraph{Deriving the Lower Bound.}
From the bounds on the communication, we can finally derive our lower bound. First observe that the lower bound in Theorem~\ref{thm:mainDelete} is trivial for $B = \lg^cN$ for any constant $c>0$ (it only claims $\Omega(1/B)$). Thus we focus on the case of $B = \lg^{\omega(1)} N$. We set $\beta = \lg^{c'} N$ for a sufficiently large constant $c'>0$. Since we assume $N \geq M^{1+\eps}$ for some constant $\eps>0$, we get that with this choice of $\beta$, we have $h = \Theta(\lg N/\lg \lg N)$. Now assume for contradiction that $\max\{t_I,t_D,t_E\} = o((1/B)\lg_{\lg N} B)$. The analysis above now shows that the priority queue $\D$ can be used to obtain a Two-Phase protocol for $\UINTa{(M h \beta^h)^4,Mh\beta^{h_v-1},M\beta^{h_v}}$ (i.e. $|X| = Mh\beta^{h_v-1}$ and $|Y|=M\beta^{h_v} = \omega(|X|)$) in which the costs are the following:
\begin{itemize}
\item In phase one, Alice sends $o(((1/B)\lg_{\lg N} B) |X|w/h) = o(|X|)$.
\item In phase one, Bob first sends $o(((1/B)\lg_{\lg N} B) hw|Y|) +o(|Y|)= o(|Y|)$ bits. He then sends 
$$
o(((1/B)\lg_{\lg N} B)|Y|Bw/\beta) = o(|Y|(w \lg_{\lg N} B)/\beta) = o(|Y|)$$
bits.
\item In phase two, Alice first sends $o(((1/B)\lg_{\lg N} B)\beta w |X|)+o(|X|) = o(|X|)$ bits. She then sends $o(((1/B)\lg_{\lg N} B) Bw|X|/h) = o(|X|\lg B)$ bits.
\item In phase two, Bob sends $o(((1/B)\lg_{\lg N} B)|Y|w/\beta) = o(|Y|/\sqrt{B})$ bits.
\end{itemize}
Thus under our contradictory assumption, we have obtained a Two-Phase protocol $P$ for 
$$
\UINTa{(M h \beta^h)^4,Mh\beta^{h_v-1},M\beta^{h_v}}
$$ of cost $C(P,\UINTa{N^4,|X|,|Y|}) = (o(|X|),o(|Y|), o(|X|\lg B), o(|Y|/\sqrt{B}))$. But this is impossible by Lemma~\ref{lem_com_lower_si}. We therefore conclude that $\max\{t_I,t_D,t_E\} = \Omega((1/B)\lg_{\lg N} B)$, completing the proof of Theorem~\ref{thm:mainDelete} and thus also Theorem~\ref{thm:main}.

To summarize the main technical idea in our proof, the first phase corresponds to probes made by child $c_k(v)$. These probes are simulated by Alice and thus she will be asking for cell contents. Since a cell address costs a factor $B$ less to specify than its contents, Alice will be very silent in phase one. In phase two, Bob is simulating the probes made by $c_{k+1}(v),\dots,c_{2+\beta}(v)$ into $c_k(v)$. Therefore, he is asking for cell contents and his communication is factor $B$ less than what Alice needs to say to reply with cell contents. Thus one can interpret the Two-Phase communication lower bound as saying: if the data structure does not read enough cells from $c_1(v)$ while processing Deletes in $c_k(v)$'s subtree (small communication in phase one), then the priority queue has not learned which Deletes in $c_k(v)$'s subtree that correspond to keys inserted in $c_1(v)$. It therefore cannot structure the Deletes into a nice data structure supporting efficient ExtractMins in $c_{2+\beta}(v)$ (phase two).

\subsection{Avoiding Deletes of Keys not in the Priority Queue}
\label{sec:avoidspurios}
Recall that we assumed above that we are allowed to perform Delete operations on keys not in the priority queue, and that such spurious Deletes have no effect. We show in the following how to modify our proof such that the lower bound is derived for a sequence of operations that never delete a key not present in the priority queue.

In the previous sections, the hard distribution is defined solely from the tree $\Tree$. The tree $\Tree$ defines $O(N)$ operations on key-priority pairs with keys from a universe of size $O(N^4)$. Here is how to modify the distribution to avoid Deletes on non-inserted elements: Before performing the updates in $\Tree$, start by inserting all keys in the universe, with priorities $\infty$ (or some priority larger than all priorities ever used by $\Tree$). Note that this gives $O(N^4)$ updates, not $O(N)$. Therefore, we perform updates corresponding to $N^3$ independently chosen trees $\Tree_1$, $\Tree_2, \dots, \Tree_{N^3}$, one after the other. Since the lower bound proved in the previous section only has $\lg N$ or $\lg B$'s and we scale the number of operations $N$ by a polynomial factor, the asymptotic lower bound does not change and we can just zoom in on one tree that performs well in expectation and do the same analysis as above. This of course creates the problem that a tree has to "clean up" to make sure that after processing a tree $\Tree_i$, all keys are again in the priority queue with priority $\infty$. This needs further modifications: Every Insert operation in a tree $\Tree$ is transformed into two operations: first a Delete on the key that we are about to insert, and then the Insert. This ensures we never insert anything already there. Every Delete operation in the tree is transformed into two operations: First the Delete operation, and then an Insert with the same key, but priority $\infty$. Finally, all extract-min leaves are augmented as follows: Extract-min leaves already re-insert all key-priority pairs extracted with a priority not matching the current level of the tree. We do not change that. But the extract-min leaves also effectively removes those keys that match in priority. We therefore add an Insert for those keys, now having priority $\infty$. In this way, Alice and Bob can still solve the set disjointness communication game, and we never delete a key not present in the priority queue. The only extra requirement this transformation introduces, is that we now must have $N^{1/4} \geq M^{1+\eps}$, i.e. $N \geq M^{4+\eps}$.

%% file: com_lower.tex
\section{Proof of the Communication Lower Bound}
\label{sec:comm}
\newcommand{\DINT}[1]{\mathcal{D}^{\mathrm{SI}}_{{#1}}}
\newcommand{\UINT}[1]{\mathcal{U}^{\mathrm{SI}}_{{#1}}}
\newcommand{\DIE}[1]{\mathcal{D}^{\mathrm{IE}}_{{#1}}}
\newcommand{\rpub}{R^{\mathrm{pub}}}
\newcommand{\rprivx}{R_A^{\mathrm{priv}}}
\newcommand{\rprivy}{R_B^{\mathrm{priv}}}
\newcommand{\rrpub}{r^{\mathrm{pub}}}
\newcommand{\rrprivx}{r_A^{\mathrm{priv}}}
\newcommand{\rrprivy}{r_B^{\mathrm{priv}}}
\newcommand{\SI}[1]{\mathsf{SetInt}_{{#1}}}
\newcommand{\IE}[1]{\mathsf{IndEq}_{{#1}}}
\newcommand{\eq}{\mathrm{eq}}
\newcommand{\ie}{\mathrm{ie}}
\newcommand{\si}{\mathrm{si}}
\newcommand{\dkl}{D_\mathrm{{KL}}}
\newcommand{\cD}{\mathcal{D}}
\newcommand{\cU}{\mathcal{U}}

In this section, we are going to prove Lemma~\ref{lem_com_lower_si}, a Two-Phase communication lower bound for the set intersection problem on uniformly random inputs.

\paragraph{Hard distribution $\UINT{U,k,l}$ for $\SI{U,k,l}$}
Alice's input $X$ is sampled from all ${U\choose k}$ subsets of $[U]$ of size $k$ uniformly at random. Independently, Bob's input $Y$ is sampled from all ${U\choose l}$ subsets of $[U]$ of size $l$ uniformly at random.

\begin{customlem}{\ref{lem_com_lower_si}}
No Las Vegas Two-Phase protocol $P$ for $\SI{U,k,l}$ can have 
\[
	C(P,\UINT{U,k,l})=(o(k),o(l),o(k\log B),o(k/B))
\]
for any $B>1$, $l\gg k$ and $U\gg l$.
\end{customlem}

 The main idea is to apply \emph{information complexity} arguments and do a chain of reductions:
\begin{enumerate}
\item (Proof of Lemma~\ref{lem_com_lower_si}) Since we are sampling $k$ elements uniformly at random from $[U]$ for Alice's input $X$, and $l$ elements for Bob's input $Y$ for $l\gg k$, we would expect that if the universe is divided into $k$ blocks evenly, then many blocks will have exactly one element in $X$, and if we further divide each block into $l/k$ buckets, many buckets will have exactly one element in $Y$. In this step, we show that if there is an efficient way to solve the problem on all uniform inputs, then there is way to solve it on inputs where every block [bucket resp.] has exactly one element in $X$ [$Y$ resp.]. 
\item (Proof of Lemma~\ref{lem_com_lowerk}) Since the $k$ blocks are disjoint, we are actually solving $k$ \emph{independent} instances of the problem where Alice has one element and Bob has $l/k$ elements, one in each bucket. We show that if such protocol exists, we can solve one instance with $1/k$ of the \emph{information cost} of the original protocol. 
\item (Proof of Lemma~\ref{lem_com_lower1}) In this step, we eliminate Phase I. Alice has only one element, so her input can be viewed as a pair: (bucket number, offset within the bucket). They just want to know if Bob's element in this bucket has exactly the same offset.
We show that since Alice and Bob do not talk too much in Phase I, there is way to fix the transcript such that conditioned on the transcript, neither Alice nor Bob have revealed too much information about their offsets. We will show that even if the bucket number of Alice's input is revealed to Bob at the end of Phase I, the problem is still hard. In this case, their Phase II is just equality testing (offsets within the bucket) with inputs close to uniformly random.
\item (Proof of Lemma~\ref{lem_com_lower_eq}) In this step, we actually ``separate'' the communication in Phase I and Phase II. Although we have already eliminated Phase I in the previous step, if we try to apply rectangle argument directly, it turns out we will just end up adding the information cost of the two phases together in the proof. Thus, before applying any real lower bound argument, we first show that if the inputs are close to uniform, there will be a large subset of the universe such that each of them occurs with decent probability. The input distribution can be viewed as a convex combination of, the uniform distribution on this smaller set (with constant weight), and some other distribution. That is, with constant probability, they are actually solving the problem on a smaller universe, but with uniform inputs. This shows that the protocol can solve the uniform input with a constant factor blow-up in the information cost.
\item (Proof of Lemma~\ref{lem_com_lower_eq_uni}) Apply the standard rectangle argument, we prove a lower bound for the uniform input case. 
\end{enumerate}

As a tool heavily used in the chain of reductions, we define the information cost of a protocol $P$ on input distribution $\mu$ as follows:

\begin{definition}
We say $I(P,\mu)=(a_1,b_1,a_2,b_2)$, if
\begin{enumerate}
\item
	$I(X;\Pi_1|Y,\rpub,\rprivy)\leq a_1$;
\item
	$I(Y;\Pi_1|X,\rpub,\rprivx)\leq b_1$;
\item
	$I(X;\Pi_2|Y,\Pi_1,\rpub,\rprivy)\leq a_2$;
\item
	$I(Y;\Pi_2|X,\Pi_1,\rpub,\rprivx)\leq b_2$,
\end{enumerate}
where $\Pi_1$ is the transcript generated by $P$ on input pair $(X, Y)$ in Phase I, $\Pi_2$ is the transcript of Phase II, $\rpub$ is the public coin, and $\rprivx$ and $\rprivy$ are the private coins used by Alice and Bob respectively.
\end{definition}

That is, in Phase I, Alice reveals at most $a_1$ bits of information, Bob reveals at most $b_1$; in Phase II, Alice reveals at most $a_2$ bits of information and Bob reveals at most $b_2$.

As in the standard communication model, information cost is always a lower bound of communication cost, i.e., $I(P, \mu)\leq C(P,\mu)$. Since we always work with product input distributions in this section, by the following fundamental lemma about communication complexity, the definition remains equivalent if all conditions on $X,Y,\rprivx,\rprivy$ are removed. 

\begin{lemma}\label{lem_trans_indep}
For production distribution on input pair $(X, Y)$, conditioned on a (prefix of a) transcript $\Pi=\pi$ and public coin flips $\rpub=\rrpub$, $(X,\rprivx)$ and $(Y,\rprivy)$ are independent.
\end{lemma}

The lemma can be proved by induction, we will omit its proof here. In the following, we start by proving a lower bound for the base problem $\mathsf{EQ}$, then we will build hardness on top of it.

\paragraph{The $\mathsf{EQ}$ problem and hard distribution $\cU_W\times \cU_W$} Alice and Bob receive two independent uniform random number from set $[W]$. Their goal for $\mathsf{EQ}$ is to determine whether the two numbers equal. 

\begin{lemma}\label{lem_com_lower_eq_uni}
For any Las Vegas (standard One-Phase) protocol $P_{\eq}$ for $\mathsf{EQ}$, let $\Pi$ be the random transcript $P_{\textrm{eq}}$ generated on $(X_{\eq},Y_{\eq})\sim \mathcal{U}_W\times\mathcal{U}_W$ with public coins $\rpub$ and private coins $\rprivx,\rprivy$, $P_{\eq}$ cannot have both
\[
	I(X_{\eq};\Pi|\rpub)\leq\frac{1}{2}\log B,
\]
and
\[
	I(Y_{\eq};\Pi|\rpub)\leq\frac{1}{2B}
\]
for any $B>1$ and $W\geq 2$.
\end{lemma}

\begin{proof}
Assume for contradiction, such protocol $P_{\eq}$ exists. Then by definition, we have
\[
\begin{aligned}
	\frac{1}{2}\log B&\geq I(X_{\eq};\Pi|\rpub) \\
	&=H(X_{\eq})-H(X_{\eq}|\Pi,\rpub) \\
	&=\E_{\pi,\rrpub}[H(X_{\eq})-H(X_{\eq}|\Pi=\pi,\rpub=\rrpub)]
\end{aligned}
\]
Note that $X_{\eq}$ is uniform, the difference inside the expectation is always non-negative. By Markov's inequality, with $\geq 1/2$ probability, $$H(X_{\eq}|\Pi=\pi,\rpub=\rrpub)\geq H(X_{\eq})-\log B.$$ Similarly, with $\geq 1/2$ probability, $$H(Y_{\eq}|\Pi=\pi,\rpub=\rrpub)\geq H(Y_{\eq})-1/B.$$ Thus, by union bound, there exist some $\pi$ and $\rrpub$ such that both conditions hold.

However, since $P_{\eq}$ is zero-error, conditioned on $\Pi=\pi$ and $\rpub=\rrpub$, the set of all values for $X_{\eq}$ that occur with non-zero probability and the set of all values for $Y_{\eq}$ that occur with non-zero probability must form a monochromatic rectangle. Thus, these two sets of values must be pairwise distinct (or pairwise equal, in which case we also have the following), implying
\[
	2^{H(X_{\eq}|\Pi=\pi,\rpub=\rrpub)}+2^{H(Y_{\eq}|\Pi=\pi,\rpub=\rrpub)}\leq W.
\]
However, we also have
\[
\begin{aligned}
	2^{H(X_{\eq}|\Pi=\pi,\rpub=\rrpub)}+2^{H(Y_{\eq}|\Pi=\pi,\rpub=\rrpub)}&\geq 2^{\log W-\log B}+2^{\log W-1/B} \\
	&=W(2^{-\log B}+2^{-1/B}) \\
	&=W\left(\frac{1}{B}+e^{-\frac{\ln 2}{B}}\right) \\
	&\geq W\left(\frac{1}{B}+1-\frac{\ln 2}{B}\right) \\
	&>W.
\end{aligned}
\]
We have a contradiction, such protocol does not exist.
\end{proof}

The following lemma shows that even if the input distribution is slightly non-uniform, the problem is still hard.

\begin{lemma}\label{lem_com_lower_eq}
Fix any two distributions $\mathcal{D}_X$ and $\mathcal{D}_Y$ over $[V]$, such that for $X_{\eq}\sim\mathcal{D}_X$ and $Y_{\eq}\sim\mathcal{D}_Y$, $H(X_{\eq})\geq \log V-1/18$ and $H(Y_{\eq})\geq \log V-1/18$. Then for any Las Vegas (One-Phase) protocol $P_{\textrm{eq}}$ for $\mathsf{EQ}$, let $\Pi$ be the random transcript $P_{\textrm{eq}}$ generated on $(X_{\eq},Y_{\eq})\sim \mathcal{D}_X\times\mathcal{D}_Y$, with public coins $\rpub$ and private coins $\rprivx,\rprivy$, then $P_{\eq}$ cannot have both
\[
	I(X_{\eq};\Pi|\rpub)\leq\frac{1}{12}\log B,
\]
and
\[
	I(Y_{\eq};\Pi|\rpub)\leq\frac{1}{12B}
\]
for any $B>1$ and $V\geq 6$.
\end{lemma}

Note that we cannot hope to increase the $1/18$ in the statement to $1$, as then $\cD_X$ and $\cD_Y$ can have completely disjoint support.

\begin{proof}
Assume for contradiction, such protocol $P_{\eq}$ exists. Then by Pinsker's inequality, 
$$\delta(\mathcal{D}_X,\mathcal{U}_V)\leq \sqrt{\frac{1}{2}\dkl(\mathcal{D}_X\|\mathcal{U}_V)}=\sqrt{\frac{1}{2}\left(\log V-H(X_{\eq})\right)}\leq 1/6.\footnote{$\delta(P,Q):=\sup_{\mathrm{event }A}|P(A)-Q(A)|$.}$$ 
Let $W_X=\left\{x: \Pr_{\cD_X}[X_{\eq}=x]\geq \frac{1}{2V}\right\}$, we must have $|W_X|\geq \frac{2}{3}V$.
Otherwise, $\cD_X$ and $\cU_V$ on the event $X_{\eq}\not\in W_X$ will have difference more than $1/6$.
Similarly, let $W_Y=\left\{y:\Pr[Y_{\eq}=y]\geq \frac{1}{2V}\right\}$, we have $|W_Y|\geq \frac{2}{3}V$. Let $W=W_X\cap W_Y$, by union bound, $|W|\geq \frac{1}{3}V$. 

The distribution $\mathcal{D}_X$ [$\cD_Y$ resp.] can be viewed as a convex combination of two distributions: the uniform distribution $\mathcal{U}_{W}$ over $W$ with weight $1/6$, and some other distribution $\overline{\cD}_X$ [$\overline{\cD}_Y$ resp.] with weight $5/6$. 
Note that such decompositions are possible whenever \[\Pr_{X_{\eq}\sim \cD_X}[X_{\eq}=x]\geq \frac{1}{6}\Pr_{X_{\eq}\sim \cU_{W}}[X_{\eq}=x]\] and \[\Pr_{Y_{\eq}\sim \cD_Y}[Y_{\eq}=y]\geq \frac{1}{6}\Pr_{Y_{\eq}\sim \cU_{W}}[Y_{\eq}=y]\]
hold for every $x$ and $y$. 

Thus, to sample an input pair $(X_{\eq}, Y_{\eq})\sim \cD_X\times \cD_Y$, we can think of it as a two-step process: first decide independently for $X_{\eq}$ [$Y_{\eq}$ resp.], whether to sample them from $\cU_{W}$ or $\overline{\cD}_X$ [$\overline{\cD}_Y$ resp.], then sample them from the corresponding distribution. Let $C_X$ [$C_Y$ resp.] be the random variable indicating the outcome of the first step (taking value $0$ indicates sampling from $\cU_{W}$). 
We are going to show that the Las Vegas protocol $P_{\eq}$ itself is too good for input pair $(X'_{\eq}, Y'_{\eq})\sim\cU_{W}\times \cU_{W}$. Let $\Pi'$ be the (random) transcript generated by $P_{\eq}$ on $(X'_{\eq},Y'_{\eq})$ with public coin $\rpub{}'$. 

We have $I(X'_{\eq};\Pi'|\rpub{}')=I(X_{\eq};\Pi|\rpub,C_X=0)$. On the one hand, $$I(X_{\eq};\Pi|\rpub,C_X)=\frac{1}{6}I(X_{\eq};\Pi|\rpub,C_X=0)+\frac{5}{6}I(X_{\eq};\Pi|\rpub,C_X=1).$$
On the other hand,
\[
	\begin{aligned}
		I(X_{\eq};\Pi|\rpub,C_X)&=I(X_{\eq},C_X;\Pi|\rpub)-I(C_X;\Pi|\rpub) \\
		&=I(X_{\eq};\Pi|\rpub)+I(C_X;\Pi|\rpub,X_{\eq})-I(C_X;\Pi|\rpub) \\
		&\leq I(X_{\eq};\Pi|\rpub).
	\end{aligned}
\]
The last inequality is using the fact that conditioned on $X_{\eq}$, $C_X$ is independent of $(\Pi,\rpub)$, and thus $I(C_X;\Pi|\rpub,X_{\eq})=0$. Therefore, $I(X'_{\eq};\Pi'|\rpub{}')\leq 6I(X_{\eq};\Pi|\rpub)\leq \frac{1}{2}\log B$. Similarly, we have $I(Y'_{\eq};\Pi'|\rpub{}')\leq 6I(Y_{\eq};\Pi|\rpub)\leq \frac{1}{2B}$. However, this is impossible for any Las Vegas protocol by Lemma~\ref{lem_com_lower_eq_uni}. 
\end{proof}

\paragraph{The $\IE{V,L}$ problem and hard distribution $\DIE{V,L}$} Alice receives $X=(F, O)$, where $F$ is a uniformly random number in $[L]$, and $O$ is a uniformly random number in $[V]$. Bob receives $Y=(Y_1,\ldots,Y_L)$, where each $Y_i$ is a uniformly random number in $[V]$. Their goal is to decide whether $O=Y_F$. 

\begin{lemma}\label{lem_com_lower1}
No Las Vegas Two-Phase protocol $P_{\ie}$ for $\IE{V, L}$ can have
\[
	I(P_{\ie},\DIE{V,L})=\left(\frac{1}{216},\frac{L}{288},\frac{\log B}{144},\frac{1}{48B}\right)
\]
for any $B>1$ and $V\geq 6$.
\end{lemma}
\begin{proof}

Assume for contradiction, such protocol $P$ exists. Let $X_{\ie}=(F,O)$ and $Y_{\ie}=(Y_1,\ldots, Y_L)$. By definition, we have
\begin{enumerate}
\item
	$I(F, O;\Pi_1|\rpub)\leq \frac{1}{216}$;
\item
	$I(Y_1,\ldots, Y_L;\Pi_1|\rpub)\leq \frac{L}{288}$;
\item
	$I(X;\Pi_2|\Pi_1,\rpub)\leq \frac{\log B}{144}$;
\item
	$I(Y;\Pi_2|\Pi_1,\rpub)\leq \frac{1}{48B}$.
\end{enumerate}
By the first inequality above and the fact that public coin is independent of the input, we have
\[
\begin{aligned}
	\frac{1}{216}&\geq H(F, O)-H(F, O|\Pi_1,\rpub) \\
	&=\E_{\pi_1,\rrpub}\left[H(F, O)-H(F, O|\Pi_1=\pi_1,\rpub=\rrpub)\right].
\end{aligned}
\]
Since $(F, O)$ is uniform, the difference inside the expectation is always non-negative. By Markov's inequality, with probability $\geq 3/4$, 
\[
H(F, O)-H(F, O|\Pi_1=\pi_1,\rpub=\rrpub)\leq \frac{1}{54}.
\]
That is, $H(F, O|\Pi_1=\pi_1,\rpub=\rrpub)\geq \log L+\log V-1/54$. Similarly, with probability $\geq 3/4$, $H(Y_1,\ldots,Y_L|\Pi_1=\pi_1,\rpub=\rrpub)\geq L\cdot \log V-L/72$. By Markov's inequality again, with probability $\geq 3/4$, $I(X;\Pi_2|\Pi_1=\pi_1,\rpub=\rrpub)\leq \frac{1}{36}\log B$ and with probability $\geq 3/4$, $I(Y;\Pi_2|\Pi_1=\pi_1,\rpub=\rrpub)\leq \frac{1}{12B}$. Thus, by union bound, there exists some $(\pi_1,\rrpub)$ such that all four inequalities hold:
\begin{enumerate}
\item
	$H(F, O|\Pi_1=\pi_1,\rpub=\rrpub)\geq \log L+\log V-\frac{1}{54}$;
\item
	$H(Y_1,\ldots,Y_L|\Pi_1=\pi_1,\rpub=\rrpub)\geq L\cdot \log V-\frac{L}{72}$;
\item
	$I(F, O;\Pi_2|\Pi_1=\pi_1,\rpub=\rrpub)\leq \frac{1}{36}\log B$;
\item
	$I(Y_1,\ldots,Y_L;\Pi_2|\Pi_1=\pi_1,\rpub=\rrpub)\leq \frac{1}{12B}$.
\end{enumerate}

That is, conditioned on $(\Pi_1,\rpub)=(\pi_1,\rrpub)$, the inputs still have high entropy and the players do not reveal too much information about their inputs in Phase II. We are going to show that it is impossible by Lemma~\ref{lem_com_lower_eq}. From now on, we fix $\pi_1$ and $\rrpub$ and condition on $\Pi_1=\pi_1,\rpub=\rrpub$.
\begin{enumerate}
\item
	By chain rule, we have
	\[
	\begin{aligned}
	H(O|\Pi_1=\pi_1,\rpub=\rrpub, F)&=H(F, O|\Pi_1=\pi_1,\rpub=\rrpub) \\
	&\quad-H(F|\Pi_1=\pi_1,\rpub=\rrpub)\\
	&\geq H(F, O|\Pi_1=\pi_1,\rpub=\rrpub)-\log L \\
	&\geq \log V-1/54.
	\end{aligned}
	\]
	By Markov's inequality, $H(O|\Pi_1=\pi_1,\rpub=\rrpub, F=f)\geq \log V-1/18$ with at least $2/3$ probability (over a random $f$). 
\item
	By sub-additivity of entropy, $$\sum_{f=1}^L H(Y_f|\Pi_1=\pi_1,\rpub=\rrpub)\geq L\cdot\log V-L/72.$$ Since each $H(Y_f|\Pi_1=\pi_1,\rpub=\rrpub)\leq \log V$, by Markov's inequality, at least $3L/4$ different $f$'s have $H(Y_f|\Pi_1=\pi_1,\rpub=\rrpub)\geq \log V-1/18$. These $f$'s must occur with high probability, since $F$ has high entropy: $H(F|\Pi_1=\pi_1,\rpub=\rrpub)\geq \log L-1/54$. In fact, if we had
\[
	\sum_{f:H(Y_f|\Pi_1=\pi_1,\rpub=\rrpub)\geq \log V-1/18} \Pr\left[F=f|\Pi_1=\pi_1,\rpub=\rrpub\right]<2/3,
\]
then by the concavity of $p\log\frac{1}{p}$, we would have
\[
\begin{aligned}
H(F|\Pi_1=\pi_1,\rpub=\rrpub)&<\frac{2}{3}\log \frac{0.75L}{2/3}+\frac{1}{3}\log \frac{0.25L}{1/3} \\
&=\frac{2}{3}\log\frac{9L}{8}+\frac{1}{3}\log\frac{3L}{4} \\
&=\log L-\frac{1}{3}\log\frac{256}{243} \\
&< \log L-1/40.
\end{aligned}
\]
Thus, by Lemma~\ref{lem_trans_indep}, $H(Y_F|\Pi_1=\pi_1,\rpub=\rrpub,F=f)\geq \log V-1/18$ with probability at least $2/3$ (over a random $f$).
\item
	By chain rule,
	\[
	\begin{aligned}
	I(O;\Pi_2|\Pi_1=\pi_1,\rpub=\rrpub,F)&=I(F,O;\Pi_2|\Pi_1=\pi_1,\rpub=\rrpub) \\
	&\quad-I(F;\Pi_2|\Pi_1=\pi_1,\rpub=\rrpub) \\
	&\leq \frac{1}{36}\log B.
	\end{aligned}
	\]
	By Markov's inequality, $I(O;\Pi_2|\Pi_1=\pi_1,\rpub=\rrpub,F=f)<\frac{1}{12}\log B$ with probability at least $2/3$ (over a random $f$). 
\item
	By Lemma~\ref{lem_trans_indep}, we have
	\[
	\begin{aligned}
	I(Y_F;\Pi_2|\Pi_1=\pi_1,\rpub=\rrpub,F=f)&=I(Y_f;\Pi_2|\Pi_1=\pi_1,\rpub=\rrpub)\\
	&\leq I(Y;\Pi_2|\Pi_1=\pi_1,\rpub=\rrpub)\\
	&\leq \frac{1}{12B}.
	\end{aligned}
	\]
\end{enumerate}

Thus, by union bound, there is some $f\in [L]$, such that
\begin{enumerate}
\setcounter{enumi}{-1}
\item
	by Lemma~\ref{lem_trans_indep}, $O$ and $Y_F$ are independent conditioned on $\Pi_1=\pi_1,\rpub=\rrpub,F=f$;
\item
	$H(O|\Pi_1=\pi_1,\rpub=\rrpub,F=f)\geq \log V-1/18$;
\item
	$H(Y_F|\Pi_1=\pi_1,\rpub=\rrpub,F=f)\geq \log V-1/18$;
\item
	$I(O;\Pi_2|\Pi_1=\pi_1,\rpub=\rrpub,F=f)\leq\frac{1}{12}\log B$;
\item
	$I(Y_F;\Pi_2|\Pi_1=\pi_1,\rpub=\rrpub,F=f)\leq\frac{1}{12B}$.
\end{enumerate}

That is, when $\Pi_1=\pi_1,\rpub=\rrpub,F=f$, Phase II of protocol $P$ can decide if $Y_F=O$ with low information cost, while both inputs $Y_F$ and $O$ have high entropy. However, Lemma~\ref{lem_com_lower_eq} asserts that it is impossible.
\end{proof}

\paragraph{Hard distribution $\DINT{U,k,l}$ for $\SI{U,k,l}$}
Partition $[U]$ into $k$ blocks $\{A_1,\ldots,A_k\}$ of size $U/k$, and partition each block $A_i$ further into $l/k$ buckets $\{B_{i,1},\ldots,B_{i,l/k}\}$ of size $U/l$. Alice's input set $X$ contains one uniformly random element from each $A_i$ for $1\leq i\leq k$, Bob's input set $Y$ contains one uniformly random element from each $B_{i,j}$ for $1\leq i\leq k$ and $1\leq j\leq l/k$. Thus, $X$ is a subset of $[U]$ of size $k$, and $Y$ is a subset of size $l$.

Note that by associating each bucket a number $Y_i$, and represent $X$ as the pair (bucket number, offset in the bucket), the $\SI{U/k,1,l/k}$ problem on $\DINT{U/k,1,l/k}$ is equivalent to the $\IE{U/l,l/k}$ problem on $\DIE{U/l, l/k}$. 

\begin{lemma}\label{lem_com_lowerk}
No Las Vegas Two-Phase protocol $P_{\si}$ for $\SI{U,k,l}$ can have
\[
	I(P_{\si},\DINT{U,k,l})=\left(\frac{k}{216},\frac{l}{288},\frac{k\log B}{144},\frac{k}{48B}\right),
\]
as a consequence, no $P_{\si}$ can have
\[
	C(P_{\si},\DINT{U,k,l})=\left(\frac{k}{216},\frac{l}{288},\frac{k\log B}{144},\frac{k}{48B}\right)
\]
for any $B>1$, $U\geq 6l$.
\end{lemma}

\begin{proof}
Assume for contradiction, such protocol $P_{\si}$ exists. We show that it can be used to solve $\IE{U/l,l/k}$ on $\DIE{U/l,l/k}$ efficiently by a standard information complexity argument.

Given $P_{\si}$, consider the following protocol $P_{\ie}$ on input pair $(X_{\ie},Y_{\ie})\sim\DIE{U/l,l/k}$:
\begin{enumerate}
\item
	Use \emph{public coin} to uniformly sample a random $I$ from the set $\{1,\ldots,k\}$.
\item
	Generate an input pair $(X, Y)$ for $\SI{U,k,l}$. Partition $[U]$ into $\{A_1,\ldots,A_k\}$ as for $\DINT{U,k,l}$. For block $A_I$, add the corresponding elements in $X_{\ie}$ and $Y_{\ie}$ to $X$ and $Y$ respectively, i.e., for $X_{\ie}=(F, O)$, Alice adds $O+(F-1)U/l+(I-1)U/k$ to $X$, and for each $Y_f$ in $Y_{\ie}$, Bob adds $Y_f+(f-1)U/l+(I-1)U/k$ to $Y$. 
	For all other blocks, they use \emph{private coin} to sample the elements for $(X,Y)$ as in $\DINT{U,k,l}$, i.e., Alice samples one uniformly random element from the block, Bob further partitions it into $l/k$ buckets and samples one uniformly random element from each of them. It is not hard to see that if $(X_{\ie},Y_{\ie})\sim \DIE{U/l,l/k}$, then $(X, Y)\sim \DINT{U,k,l}$.
\item
	Run protocol $P_{\si}$ on $(X, Y)$ to compute the intersection $X\cap Y$, and return the set $X\cap Y\cap A_I$ (with a shift of $-(I-1)U/k$).
\end{enumerate}

Since different blocks are disjoint, the set $X\cap Y\cap A_I$ tells the players whether $Y_F=O$. Thus, as long as $P_{\si}$ is zero-error (for $\SI{U,k,l}$), $P_{\ie}$ is also a zero-error protocol (for $\IE{U/l,l/k}$). To analyze its information cost, let $X^{(i)}=X\cap A_i$, we have
\[
\begin{aligned}
	I(X_{\ie};\Pi_1|\rpub)&=\frac{1}{k}\sum_{i=1}^k I(X_{\ie};\Pi_1|I=i,\rpub(P_{\si})) \\
	&=\frac{1}{k}\sum_{i=1}^k I(X^{(i)};\Pi_1|I=i,\rpub(P_{\si})) \\
	&=\frac{1}{k}\sum_{i=1}^k I(X^{(i)};\Pi_1|\rpub(P_{\si})) \\
	&\leq \frac{1}{k}I(X;\Pi_1|\rpub(P_{\si})) \\
	&\leq \frac{1}{216},
\end{aligned}
\]
where $\rpub(P_{\si})$ is the random coins used in $P_{\si}$.

Similar arguments also prove that $I(Y_{\ie};\Pi_1|\rpub)\leq \frac{l}{288k}$, $I(X_{\ie};\Pi_2|\Pi_1,\rpub)\leq \frac{1}{144}\log B$ and $I(Y_{\ie};\Pi_2|\Pi_1,\rpub)\leq \frac{1}{48B}$. But by Lemma~\ref{lem_com_lower1} (setting $V=U/l$ and $L=l/k$), such protocol does not exist. 
\end{proof}

Finally, we are ready to prove Lemma~\ref{lem_com_lower_si}, a communication lower bound for $\SI{U,k,l}$ on uniformly random inputs $(X, Y)\sim \UINT{U,k,l}$.

\begin{proof}[Proof of Lemma~\ref{lem_com_lower_si}]
Assume for contradiction, such protocol $P$ exists. We are going to design a ``too good'' protocol for $\SI{U/9,k/9,l/9}$ and input distribution $\DINT{U/9,k/9,l/9}$.

Partition $[U]$ into $l$ buckets of size $U/l$ each. Bob's input set $Y$ is sampled uniformly from all ${U\choose l}$ subsets of $[U]$ of size $l$. We can think of the sampling of $Y$ being a two-step process: first decide the number of elements of $Y$ in each bucket, then sample a uniformly random subset of this size from the bucket. Intuitively, with high probability, there will be a constant fraction of the buckets that are going to have exactly one element after the first step, which is stated as the following observation:
\begin{observation}\label{obs_one_ball}
The probability that at least $l/3$ buckets have exactly one element is $1-o(1)$.
\end{observation}

In order to focus on proving our communication lower bound, we defer the proof of the observation to the end of the section. 
Assume this happens, then group these $l/3$ buckets into $k/3$ blocks, each of which consists of $l/k$ buckets. We now think of $X$ also being sampled by a two-step process: first decide the number of elements of $X$ in each block (which also determines the number of elements outside these $l/3$ buckets), then sample uniformly random sets of the corresponding sizes from the blocks (and from the elements not in any block). A similar argument as Observation~\ref{obs_one_ball} shows that the probability that at least $k/9$ blocks that are going to have exactly one element with probability $1-o(1)$. Call these $k/9$ blocks the \emph{good blocks}, and their $l/9$ buckets the \emph{good buckets}.

Consider the following protocol $P_{\si}$ for $\SI{U/9,k/9,l/9}$ on input $(X_{\si}, Y_{\si})\sim \DINT{U/9,k/9,l/9}$:
\begin{enumerate}
\item
	Generate $(X, Y)\sim \UINT{U,k,l}$ in the following way. Complete the first steps of sampling $X$ and $Y$ as described above using \emph{public coin} (which fails with $o(1)$ probability). If it does not fail, for the $k/9$ good blocks and their $l/9$ good buckets, embed $X_{\si}$ and $Y_{\si}$ into them. For all other blocks and buckets, complete the second steps of sampling. Thus, when $(X_{\si}, Y_{\si})\sim \DINT{U/9,k/9,l/9}$, we have $(X, Y)\sim \UINT{U,k,l}$. 
\item
	If the sampling fails, repeat Step 1. Otherwise, run $P$ on $(X, Y)$, and return the intersection in the good blocks.
\end{enumerate}

It is not hard to verify that $P_{\si}$ solves $\SI{U/9,k/9,l/9}$ with zero-error and communication cost $(o(k),o(l),o(k\log B),o(k/B))$. By Lemma~\ref{lem_com_lowerk}, such protocol does not exist, and we have a contradiction.
\end{proof}

\begin{proof}[Proof of Observation~\ref{obs_one_ball}]
Let $\xi_i$ be the random variable indicating whether $i$-th bucket has exactly one element of $Y$ ($\xi_i=1$ indicates YES). Let $\xi=\sum_{i=1}^l \xi_i$. Thus, by Stirling's formula,
\[
\begin{aligned}
\E[\xi_i]&=\Pr[\xi_i=1] \\
&=\frac{{U-U/l \choose l-1}\cdot U/l}{{U\choose l}} \\
&\sim\frac{(U/l)^{l-1}(1-U/l)^{U-U/l-l+1}}{(U/l)^l(1-U/l)^{U-l}}\cdot \frac{U}{l} \\
&=(1-U/l)^{-U/l+1} \\
&\sim e^{-1}
\end{aligned}
\]
By linearity of expectation, $\E[\xi]\sim e^{-1} l$. On the other hand,
\[
\begin{aligned}
\E[\xi_i\xi_j]&=\Pr[\xi_i=1\wedge \xi_j=1] \\
&=\frac{{U-2U/l\choose l-2}\cdot (U/l)^2}{{U\choose l}} \\
&\sim \frac{(U/l)^{l-2}(1-U/l)^{U-2U/l-l+2}}{(U/l)^l(1-U/l)^{U-l}}\cdot\frac{U^2}{l^2}\\
&=(1-U/l)^{-2U/l+2} \\
&\sim e^{-2}.
\end{aligned}
\]
Thus, we have
\[
\begin{aligned}
\Var[\xi]&=\E[\xi^2]-(\E[\xi])^2 \\
&=\sum_{i=1}^l\E[\xi_i^2]+2\sum_{i<j}\E[\xi_i\xi_j]-(1+o(1))e^{-2}l^2 \\
&=(1+o(1))e^{-1}l+l(l-1)(1+o(1))e^{-2}-(1+o(1))e^{-2}l^2 \\
&=o(l^2).
\end{aligned}
\]
By Chebyshev's inequality, we have $\Pr[\xi<l/3]\leq o(1)$, which proves the observation. 
\end{proof}

%% file: appendix.tex
\section{Internal Memory DecreaseKeys}
\label{sec:internaldecrease}
In the following, we show how we one can support the DecreaseKey operation at no loss of efficiency for Las Vegas randomized and amortized internal memory data structures when the keys can be efficiently hashed (i.e. one can construct a hash table with expected amortized $O(1)$ time lookups).
We first describe it in the case where keys are never re-inserted after having been extracted. 

Given a priority queue supporting only Insert and ExtractMin, we modify it to support DecreaseKeys. Let $t_I(N)$ be its expected amortized insert time when storing $N$ elements and let $t_E(N)$ be its expected amortized ExtractMin time. We start by supporting DecreaseKeys in expected amortized $O(t_I(N')+t_E(N'))$ time per operation over a sequence of $N'$ operations if we are promised that a key is never re-inserted into the priority queue after having been extracted. This is done as follows: We keep a counter $C$ on the number of operations performed on the priority queue. The counter $C$ is initialized to $0$, and on each operation we increment it by $1$. In addition to $C$, we maintain a priority queue $\D$ supporting only Insert and ExtractMin, as well as a hash table $H$. The hash table $H$ will store the keys that have already been extracted. In more detail, on an Insert($k,p$), we create the key $k \circ C$, i.e. we augment the key with the time of the operation and assume that the data structure treats two keys $k \circ C$ and $k \circ C'$ as distinct keys. We then insert $(k \circ C,p)$ in $\D$ by running its Insert operation. On a DecreaseKey($k,p$), we create the key $k \circ C$ and run Insert($k \circ C, p$) on $\D$. On ExtractMin, we run ExtractMin on $\D$. This returns a pair $(k \circ C_k, p)$, where $C_k$ is the time when $k$ either underwent an Insert or a DecreaseKey. We then do a lookup in $H$ for $k$. If $k$ is not in $H$, we simply return $(k,p)$ as our result and insert $k$ in $H$. If $k$ is in $H$, we know the key has already been extracted with a smaller priority. We then invoke ExtractMin on $\D$ again until we find a key-priority pair that have not yet been extracted.

Since we assume hash tables can be implemented with $O(1)$ expected amortized time operations, and each DecreaseKey operation causes one Insert and one ExtractMin, the size of $\D$ is thus $O(N')$ where $N'$ is the number of operations performed on the priority queue. Therefore the expected amortized cost of the DecreaseKey is $O(t_I(N')+t_E(N'))$.

If we assume keys can be re-inserted after having been extracted, we can store an additional hash table with the last time at which an Insert on a given key was performed. Then when running ExtractMin and seeing a pair $(k \circ C_k,p)$, we also check the value of $C_k$ to see if it is lower than the last time $k$ was inserted. In that case, we also throw away the pair and run another ExtractMin on $\D$.

Finally, it remains to lower the expected amortized time from $O(t_I(N')+t_E(N'))$ to $O(t_I(N)+t_E(N))$ where $N$ is the number of elements actually present in the priority queue at the time of the DecreaseKey. This is done using the standard global rebuilding technique: Initialize $N_0$ to a constant. After $N_0$ operations on the priority queue, all elements are extracted using the ExtractMin operation. They are put on a list $L$. We then throw away the data structure and initialize a new one from scratch. We then re-insert all elements of $L$ in the new data structure and set $N_0$ to $|L|/2$. Since the number of elements in the priority queue is always within a constant factor of the number of operations performed since the last rebuild, the expected amortized running time is $O(t_I(N)+t_E(N))$ as claimed. This trick also ensures that Insert and ExtractMin still run in time $O(t_I(N))$ and $O(t_E(N))$ respectively.